\documentclass[12pt]{article}

\usepackage{epsfig}

\usepackage{amssymb}
\usepackage{amsfonts}

\usepackage{color}
 
\def\be{\begin{equation}}
\def\ee{\end{equation}}
\def\ba{\begin{array}{c}}
\def\ea{\end{array}}

\def\ben{$$}
\def\een{$$}

\newcommand{\bea}{\begin{eqnarray}}
\newcommand{\eea}{\end{eqnarray}}

\newcommand{\kt}{\rangle}

\newtheorem{thm}{Theorem}

\newtheorem{lemma}[thm]{Lemma}

\newtheorem{conj}[thm]{Conjecture}
\newenvironment{proof}{\noindent
 {\bf Proof.}}{\hfill$\square$\vspace{3mm}\endtrivlist}

\begin{document}

\begin{center}

{\Large \bf

 Unitary unfoldings of
 Bose-Hubbard
 exceptional point with and
 without particle number conservation

}

\vspace{0.8cm}

  {\bf Miloslav Znojil}

\vspace{0.2cm}

The Czech Academy of Sciences, Nuclear Physics Institute,

 Hlavn\'{\i} 130,
250 68 \v{R}e\v{z}, Czech Republic

\vspace{0.2cm}

 and

\vspace{0.2cm}

Department of Physics, Faculty of Science, University of Hradec
Kr\'{a}lov\'{e},

Rokitansk\'{e}ho 62, 50003 Hradec Kr\'{a}lov\'{e},
 Czech Republic

\vspace{0.2cm}

{e-mail: znojil@ujf.cas.cz}

\end{center}

\newpage

\section*{Abstract}

Conventional non-Hermitian but
${\cal PT}-$symmetric three-parametric
Bose-Hubbard
Hamiltonian $H(\gamma,v,c)$
represents a
quantum system of $N$ bosons, unitary
only
for parameters $\gamma$,$v$ and $c$ in
a
domain ${\cal D}$.
Its boundary $\partial {\cal D}$ contains
an
exceptional point of order $K$ (EPK, $K=N+1$)
at $c= 0$ and $\gamma = v$, but
even at the smallest non-vanishing
parameter $c \neq 0$ the spectrum
of
$H(v,v,c)$
ceases to be real, i.e.,
the system ceases to be
observable.
In our paper the question is inverted:
All of the stable, unitary and
observable Bose-Hubbard quantum systems are sought which
would lie
close to the
phenomenologically most interesting
EPK-related dynamical regime.
Two different families of such systems are found. Both
of them are characterized by the perturbed Hamiltonians
$\mathfrak{H}(\lambda)=H(v,v,0)+\lambda\,{\cal V}$
for which
the unitarity and stability of the system is guaranteed.
In the first family
the number $N$ of bosons is assumed conserved while
in the second family such an assumption is relaxed.
Our main attention is paid to an
anisotropy
of the
physical Hilbert space
near the EPK extreme.
We show that it is reflected
by a specific,
operationally realizable structure of
perturbations $\lambda\,{\cal V}$
which can be considered small.

\subsection*{Keywords}
.

quantum mechanics,

quasi-Hermitian observables,

dynamics near exceptional points,

Bose-Hubbard Hamiltonians,

stability-guaranteeing perturbations,

\newpage

\section{Introduction\label{introduction}}

The recent growth of popularity
of the study of Schr\"{o}dinger-type evolution equations
 \be
 {\rm i}\frac{d}{dt} \psi(t) = H\,\psi(t)
 \label{setId}
 \ee
containing non-Hermitian generators {\it alias\,} Hamiltonians
$H \neq H^\dagger$ was
reflected
by
the publication of
dedicated books
\cite{Nimrod,book}.
The choice appeared relevant not only in quantum physics
but also far beyond this area \cite{Carlbook,Christodoulides}.
One of the ``hidden''
roots of the appeal of the innovative non-Hermitian
models may look like a
paradox: In these models one can weaken
the not always
desirable robust stability
of the systems controlled by
self-adjoint Hamiltonians.
{\it Pars pro toto}, let us recall the most impressive
illustration
in classical optics where due to the non-Hermiticity of $H$ one is
even
able to stop the light, in principle at least \cite{stoplight}.

In the language of mathematics
the main formal key to similar
innovations of phenomenology may be seen in the
non-Hermiticity-mediated
accessibility
of the eigenvalue degeneracies called exceptional points
(EP, \cite{Kato}). This
feature
opened
many new areas of research in physics (cf.
\cite{Carlbook,Christodoulides} or
the
sample of references listed in our preceding
paper~\cite{passage}).

In 2008, Graefe et al \cite{Uwe} followed the trend.
A turn of their attention to a non-Hermitian
version
of the popular
Bose-Hubbard Hamiltonian (cf. also \cite{[38],[37],zaUwem})
enabled them to enrich, in particular, the
toy-model-mediated understanding
of the Bose-Einstein-condensation phenomenon \cite{Cart}.
In the context of mathematics they complemented the traditional
numerical diagonalization approaches to the model
\cite{cinaj,cinajb} by pointing out the advantages of
the use of dedicated versions of
perturbation theory \cite{Puiseux,mech}.

Today, twelve years later, we intend to
revitalize the initiative of \cite{Uwe}.
We will propose two extensions of the currently studied
non-Hermitian Bose-Hubbard-type Hamiltonians.
In the first one we will keep the number of bosons conserved.
Indeed, such an assumption is still popular, mainly
in the light of the
traditional
experimental as well as theoretical
role played by the Bose-Hubbard model
in the description of
phase transitions in the systems of
ultracold spin-zero atoms
confined by an optical lattice
\cite{Matthew,Atland,ultrac}.
Still, at present, the relevance of such a
conventional
requirement
seems weakened by the recent
shift of study of similar  models
(and, in particular, of their manifestly non-Hermitian versions)
to the dynamical regime controlled by
exceptional points, especially in optics and photonics
\cite{Christodoulides,Alou}.
In the second half of our present paper, therefore, we shall
omit the requirement
of the conservation of the number of bosons
as, after such an innovation of applications ranging up to field theory
\cite{Atland}, over-restrictive.

For all of these purposes,
in a formal parallel to paper \cite{Uwe},
we intend to use
and, occasionally, amend the techniques of perturbation theory.
Still,
the physical scope of our
present study will be narrower.
In place of the open-quantum-system
setup of paper \cite{Uwe}
characterized, basically, by the
Feshbach-inspired
effective Hamiltonians \cite{Nimrod,Feshbach},
we shall restrict
attention to the mere unitary,
closed-quantum-system scenarios
characterized, first of all,
by the full, unreduced information about the dynamics.

Even when
staying inside the
unitary-evolution framework
in which the energies are real
we will keep in mind the
warnings coming from rigorous mathematics
\cite{ATbook}. Thus, we will
accept, as often as possible, the
bounded-operator
constraints as recommended in Refs.~\cite{Geyer,SIGMA}.
We will also mostly employ the terminology
used in these references,
although
we will also occasionally use some
slightly misleading but still
sufficiently well understood popular
abbreviations like
``non-Hermitian operators''.
After all, many of the similar
terminological conventions and ambiguities
were already discussed and sufficiently thoroughly clarified
elsewhere
\cite{ali,MZbook,NIP}.

\section{Conventional non-Hermitian Bose-Hubbard model}

In a way explained in \cite{Uwe,[38],[37],zaUwem},
one of the fairly realistic
descriptions of the so called
Bose-Einstein condensation phenomenon
is provided
by the specific three-parametric bosonic Hamiltonian
\begin{equation} \label{Ham1}
  H (\gamma,v,c)= -{\rm i} \gamma
  \left(a_1^{\dagger}a_1 - a_2^{\dagger}a_2\right) +
  v\left(a_1^{\dagger}a_2 + a_2^{\dagger}a_1\right) +
  c\,H_{int}\,,
  \ \ \ \
  H_{int}=
  \frac{1}{2}
  \left( a_1^{\dagger}a_1 - a_2^{\dagger}a_2\right)^2
\end{equation}
where, for two modes taken into account,
the symbols
$a_1$,
$a_2$ and $a_1^\dagger$,
$a_2^\dagger$ represent
the respective annihilation and
creation operators.
The value of coupling constant
$c$ controls the strength of
the boson-boson interaction inside a double-well potential
(we shall often consider just the interaction-free
limit $c\to 0$  in what
follows).
Parameter
$v$ measures the intensity of the tunneling
through the barrier (for convenience we shall scale it to one)
while the tunable real
quantity
$2\gamma$
stands for an
imaginary part of the on-site
bosonic-energy difference \cite{Uwe,zaUwem}.


\subsection{Matrix representation of
Hamiltonian\label{para2p1}}

The Bose-Hubbard (BH) model of Eq.~(\ref{Ham1})
is conservative in the sense that its Hamiltonian
commutes with the number operator
 \be \label{Num1}
 \widehat{N}=a_1^{\dagger}a_1 + a_2^{\dagger}a_2\,.
 \ee
For the sake of simplicity we will set
$c=0$
(meaning that the mutual interaction
between bosons is neglected) and $v=1$
(reflecting just the choice of units)
almost everywhere in what follows.
We will also make use of the bases
defined in \cite{Uwe}:
In place of operators (\ref{Ham1})
we will work with their suitable matrix
representations.
In particular, for the systems in which
the number of bosons $N$ is conserved,
Hamiltonian (\ref{Ham1})
may be given the
block-diagonal infinite-dimensional matrix form
 \be
 H_{(BH)} (\gamma)=
 H_{(BH)}^{(2+3+\ldots)} (\gamma)=
   \left (
 \begin{array}{cccc}
 H^{(2)}_{(BH)}(\gamma)&\ \ \ \ 0&0&\ldots\\
 0&H^{(3)}_{(BH)}(\gamma)&0& \\
 0&\ \ \ \ 0&H^{(4)}_{(BH)}(\gamma)&\ddots\\
 \vdots&\ \ \ \ \ &\ddots\ \ \ \ \ &\ddots
 \ea
 \right )\, \,
   \label{geneve}
 \ee
containing separate fixed$-N$ sub-Hamiltonians
 \be
 H^{(2)}_{(BH)}(\gamma)=
 \left[ \begin {array}{cc} -i{\it \gamma}&1
 \\{}1&i{\it
 \gamma}
 \end {array} \right]\,, \ \ \ \ \
H^{(3)}_{(BH)}(\gamma)=\left[ \begin {array}{ccc} -2\,i\gamma&
\sqrt{2}&0\\{}\sqrt{2}&0&
\sqrt{2}\\{}0&\sqrt{2}&2\,i\gamma\end {array}
\right]\,\ \ldots\,.
  \label{3wg}
 \ee
In a historical perspective
the oldest versions of the Bose-Hubbard models
were based on the purely imaginary choices of $\gamma$
(making all of the matrices (\ref{3wg}) Hermitian,
i.e., mathematically more user-friendly).
This helped to simulate, first of all,
the superfluid-insulator transitions \cite{Matthew}.
In contrast, the present preference of the
real-valued parameters $\gamma$
may be perceived as
mathematically less elementary but phenomenologically more promising,
especially because in contrast to the Hermitian case,
the related exceptional points may now be
reached (possibly even in an experiment).
In the language of experimental
physics this means, therefore, that such a complementary choice
opens the possibility
of reaching a quantum phase transition of the type simulating the
Bose-Einstein condensation phenomenon \cite{Uwe}.
Remarkably enough, even the instant $\gamma=0$ of
transition between the real and imaginary $\gamma$s
can be given a specific physical interpretation of a
broken-Hermiticiy quantum phase transition
(see \cite{sible} for details).

The authors of Refs.~\cite{Uwe,[38],[37],zaUwem,ngt5}
offered a number of arguments
showing that the choice of real $\gamma$ may be given a
phenomenologically consistent,
experiment-oriented meaning.
They
found, in particular,
that every
element $H^{(K)}_{(BH)}(\gamma)$ of the series
of sub-Hamiltonians (\ref{3wg})
with $K = N+1$  can be assigned the
$K-$plet of
closed-form energy eigenvalues,
 \be
 E_{n}^{(K)}(\gamma)
 =(1-\gamma^2)^{1/2}\,(1-K+2n)\,,
 \ \ \ \ \ n=0,1, \ldots,  K-1\,
 \label{spektrade}
 \ee
(see formula Nr. 28 and picture
Nr. 1 in  Ref.~\cite{Uwe}). In {\it loc. cit.}
we also find that
in
(\ref{3wg}),
only the diagonal matrix
elements would be changed
after introduction of
interaction $H_{int}$
of Eq.~(\ref{Ham1}).

\subsection{Exceptional point}

From the point of view of
traditional phenomenological applications of the
Hermitian versions of Bose-Hubbard model \cite{Matthew,Atland,ultrac},
the number $N$ of bosons in the
system was always fixed and given in advance.
After a shift of attention
to the non-Hermitian alternative of the model,
the conservation of the number of particles
was still considered useful, mainly for formal reasons.
Indeed, at a fixed $N=K-1$
(and in the present limit $c \to 0$ of course)
one easily determines the values of energies (\ref{spektrade})
and concludes that they
remain all real if an only if $\gamma^2 \leq 1$,
and non-degenerate unless
$\gamma^2 = 1$. As long as $N < \infty$,
the necessary mathematics remains elementary,
showing only that in  the two end-of-unitarity limits
$\gamma \to \pm 1$,
the limiting sub-Hamiltonians
$H^{(K)}_{(BH)}(\pm 1)$
cease to be diagonalizable. At all of the
submatrix dimensions
the energy (sub)spectra become $K-$times degenerate,
$\lim _{\gamma \to \pm 1}E_{n}^{(K)}(\gamma)=0$.
The degeneracy
applies also to the related $K-$plet of
eigenvectors (see the detailed proof in \cite{Uwe}).
Thus, the two special values of $\gamma = \pm 1$
acquire the status of the Kato's \cite{Kato} exceptional
point of order $K$ (EPK).

The picture of physics becomes different when one
recalls the full, infinite-dimensional-matrix Hamiltonian (\ref{geneve})
and when one tentatively admits the
existence of perturbations violating
the commutativity of Hamiltonian with
operator $\widehat{N}$ of Eq.~(\ref{Num1}).
Formally speaking, a part of the spectrum of operator
$H (\gamma,1,0)$ might then suddenly become infinitely degenerate
after perturbation.

Naturally, this would be an exciting, entirely new
mathematical phenomenon. Moreover,
in both of the EPK limits $\gamma \to \pm 1$, both of the corresponding
exceptional points might then also become infinitely degenerate.
This could certainly open a number of new questions
ranging from theoretical and experimental physics up to technology and
applications.

The latter observation was a key motivation of our present study.
For the sake of simplicity
let us now start the analysis by
considering just one of the two EPKs, say,
the positive one with $\gamma^{(EPK)} = 1$.
At such a parameter
the diagonalization
of submatrix $H^{(K)}_{(BH)}(1)$
is to be replaced by making it
similar
to Jordan matrix,
 \be
   H^{(K)}_{(BH)}(\gamma^{(EPK)})
 ={Q^{(K)}}\,
 J^{(K)}
 \left (\eta\right )\, \left [{Q^{(K)}}
 \right ]^{-1}\,,
   \label{JBK}
 \ee
 \be
 J^{(K)}(\eta)=\left (
 \begin{array}{ccccc}
 \eta&1&0&\ldots&0\\
 0&\eta&1&\ddots&\vdots\\
 0&0&\eta&\ddots&0\\
 \vdots&\ddots&\ddots&\ddots&1\\
 0&\ldots&0&0&\eta
 \ea
 \right )\,
   \label{JBKz}
 \ee
where, in our case, we have
$\eta=0$. Symbol
${Q^{(K)}}$
denotes the transition  matrix with the known
closed form given in \cite{passage} and forming
the sequence
 \be
 Q^{(2)}=
 \left[ \begin {array}{cc} -i&1\\{}1&0\end {array}
 \right]\,,\ \ \ \
Q^{(3)}=\left[ \begin {array}{ccc} -2&-2\,i&1\\{}-2\,i\sqrt {
2}&\sqrt {2}&0\\{}2&0&0\end {array} \right]\,, \ \ \ \ldots \,.
  \label{topp3}
 \ee
Although the
left-hand-side EP limit of the Hamiltonian
in Eq.~(\ref{JBK})
is
formally
defined by
prescription~(\ref{geneve}), it
ceased to represent an acceptable generator
of unitary evolution in quantum mechanics
because
such a matrix is not
diagonalizable anymore.


\subsection{Unitarity-breaking perturbations at a fixed $N=K-1$}

In a
vicinity
of the manifestly unphysical
exceptional-point-associated operator $H_{(BH)}(1)$
there may exist its mathematically well defined
and phenomenologically useful
perturbed descendants
 \be
 \mathfrak{H}(\lambda) = H_{(BH)}(1) + \lambda\,\mathcal{V}\,.
 \label{perha}
 \ee
We intend to show
that after the specification of an appropriate,
quantum-theoretically consistent
class of perturbations, operators (\ref{perha})
may really
re-acquire
the necessary diagonalizability
(i.e., a formal compatibility with quantum theory)
as well as many new and attractive descriptive
features.

At any fixed
number of bosons $N$
and/or
superscripted matrix dimension $K=N+1$ the
perturbed  Hamiltonian of Eq.~(\ref{perha}) degenerates to the
mere finite-dimensional
$K$ by $K$ matrix
 \be
 \mathfrak{H}^{(K)}(\lambda) = {Q^{(K)}}\,
 J^{(K)}
 \left (\eta\right )\, \left [{Q^{(K)}}
 \right ]^{-1} + \lambda\,\mathcal{V}^{(K)}\,.
 \label{Kperha}
 \ee
In 2008, attention to the related spectral problem was
attracted by Graefe et al
\cite{Uwe}. In their
paper they decided to
study
some of the perturbation-theoretical aspects of
the
realistic
as well as mathematically friendly
Bose-Hubbard Hamiltonian~(\ref{Ham1}).
In one of the dynamical scenarios
of their interest they identified the perturbation
term $\lambda\,\mathcal{V}^{(K)}$ of Eq.~(\ref{Kperha})
with the difference of matrix operators
 $
 H(\gamma,1,c)-H(1,1,0)
 $ at a fixed $N=K-1$. This enabled them to
reveal that
the growth of the strength $c>0$
diminishes, as a rule, the interval ${\cal D}$ of admissible
$\gamma$s inside which the
spectrum remains real and observable.
Subsequently
they
restricted attention to a
$\gamma=v$ subset of the special perturbations
 \be
 \lambda\,\mathcal{V}^{(K)}_{(BH)}=
 H(1,1,c)-H(1,1,0)
 \ee
which enabled them to identify the measure of the size of perturbation
$\lambda$ directly
with the  boson-boson interaction strength $c \neq 0$.
The conclusion was that at {\em any\,} nonvanishing perturbation
strength $c \neq 0$
an abrupt breakdown of the reality of the
perturbed spectrum
is inevitable.
In other words one can say that in the vicinity
of the EPK extreme
the choice of perturbation $ H_{int}$ as made in Eq.~(\ref{Ham1})
has been found incompatible with
the unitarity of the system.
In this specific dynamical regime, indeed, the conventional BH model
is not suitable for our purposes as it
only admits
the open quantum system probabilistic interpretation.

\section{Modified non-Hermitian Bose-Hubbard models}

In the present closed-quantum-system setting,
a modification of the Hamiltonian is needed.
Thus, our perturbed
Hamiltonians (\ref{perha}) will represent,
strictly speaking,
a modified, non-BH family of certain new, amended,
non-Hermitian
(i.e., more precisely, quasi-Hermitian \cite{Geyer})
but still stable and
strictly unitary
BH-type quantum systems possessing the real energy spectra.

\subsection{Perturbations conserving the number of bosons}


The block-diagonal matrix structure of
Hamiltonian (\ref{geneve}) can be interpreted as
an infinite degeneracy of energy levels (\ref{spektrade})
with respect to the number of bosons $N=K-1$.
Although such an approach looks rather formal, it
becomes relevant immediately after
the conservation of the number of bosons
happens to be broken, say, with an intention
of making the model more realistic.
A deeper understanding of such an option
[i.e., of the consequences of the possible
emergence of non-vanishing off-diagonal submatrices
in Eq.~(\ref{geneve})]
has in fact been one of the key
questions which motivated our present study.

In a preparatory step towards such a generalization
of the model it is obvious that also the separate transition
matrices (\ref{topp3})
may be inserted in the definition of a global,
infinite-dimensional
block-diagonal transition matrix,
 \be
 {\cal Q} =
 {\cal Q}^{(2+3+\ldots)} =
    \left (
 \begin{array}{cccc}
 Q^{(2)}_{}&\ \ \ \ 0&0&\ldots\\
 0&Q^{(3)}_{}&0& \\
 0&\ \ \ \ 0&Q^{(4)}_{}&\ddots\\
 \vdots&\ \ \ \ \ &\ddots\ \ \ \ \ &\ddots
 \ea
 \right )\,. \,
   \label{neve}
 \ee
Introducing
the infinite-dimensional direct-sum
generalization
of the single Jordan matrix we only have to
keep all of the limiting energy arguments equal,
 \be
 {\cal J}=
 {\cal J}^{(2+3+\ldots)}(\eta)=
   \left (
 \begin{array}{cccc}
 J^{(2)}_{}(\eta)&\ \ \ \ 0&0&\ldots\\
 0&J^{(3)}_{}(\eta)&0& \\
 0&\ \ \ \ 0&J^{(4)}_{}(\eta)&\ddots\\
 \vdots&\ \ \ \ \ &\ddots\ \ \ \ \ &\ddots
 \ea
 \right )\,.
 \label{jagen}
 \ee
This immediately leads to the full-space generalization
 \be
  H_{(BH)}(1)
 ={\cal Q}\,
{\cal J}(\eta)
 \, {\cal Q}^{-1}\,
   \label{JB}
 \ee
of the $K-$th-subspace similarity relation~(\ref{JBK}).

It should be added that the
confirmation of the
EP-related
non-diagonalizability status of matrix $H_{(BH)}(1)$
was only rendered feasible
by the exact, non-numerical tractability of its analysis.
Indeed, it is well known
\cite{ngt5} that such an EP-singularity status is
fragile and highly
sensitive to small random perturbations.
In the numerically represented models,
in particular, the
random perturbations are always present
due to the round-off errors.
This makes
the exact,
non-numerical tractability of the
conventional Bose-Hubbard Hamiltonian (\ref{Ham1})
in its
EP limit
$\gamma \to 1$
particularly rare and important.

\subsection{Perturbations which do not violate the unitarity\label{3p2}}

Hamiltonian (\ref{Kperha})
is
isospectral  to its partner matrix
 $$
 {H}^{(K)}(\lambda) =\left [{Q^{(K)}}
 \right ]^{-1}\, \mathfrak{H}^{(K)}(\lambda)\,{Q^{(K)}}\,.
 $$
Parameter $\lambda$
enters its perturbation-theory decomposition
 \be
 {H}^{(K)}(\lambda) =
 J^{(K)}(0)
  + \lambda\, W^{(K)}\,,\ \ \ \ \
  W^{(K)}=\left [{Q^{(K)}}
 \right ]^{-1}\,\mathcal{V}^{(K)}\,{Q^{(K)}}\,.
 \label{Aperha}
 \ee
Under two different philosophies,
the related EPK-related (i.e., perturbed-BH) eigenvalue problem
 \be
 {H}^{(K)}(\lambda)
 \,|\Psi_n(\lambda)\kt = E_n(\lambda)\,|\Psi_n(\lambda)\kt\,,
 \ \ \ \ n = 0, 1, \ldots, K-1\,
 \label{se17}
 \ee
has been studied in
our two older papers \cite{admissible,corridors}.
Let us now briefly recall these results.

In the former paper we
followed the older
methodical recommendation of review \cite{Geyer} so that
we
assumed that the matrix of perturbation $ W^{(K)}$
acting in the most conventional Hilbert space
${\cal K}={\cal C}^K$
is $\lambda-$independent and bounded,
 \be
  W^{(K)}\in {\cal B}({\cal K})\,.
  \label{ods}
  \ee
Under this assumption we demonstrated that, in general,
the spectrum ceases to be real
even at the smallest non-vanishing couplings $\lambda \neq 0$.
This just generalized the observations
made, in the conventional
open-system BH context, by the authors of Ref.~\cite{Uwe}.
Still, from our present, different,
closed-quantum-system point of view,
constraint (\ref{ods})
must be declared insufficient,
not guaranteeing that the
perturbations would be theoretically consistent
and experimentally
realizable.

These conclusions inspired our subsequent study
\cite{corridors}. We inverted there the question, searching
for a strengthening of the
admissibility constraint
(\ref{ods}) beyond the BH framework.
The answer has been found and based on introduction of
$\lambda-$dependent
matrices of perturbations,
 \be
 W^{(K)}=W^{(K)}(\lambda)\in {\cal B}({\cal K})\,.
 \label{ladep}
 \ee
After such an enhancement of flexibility of the constraint
and after its appropriate further amendment,
perturbed
Hamiltonians were made observable, keeping
the states inside
a preselected Hilbert
space ${\cal K}$.

As long as the required ultimate amendment
of the theory is a rather technical matter,
interested readers may find its detailed outline
in Appendix A below.
Here, let us only formulate the
final result, recalling a suitable reparametrization of
$\lambda = 1/\Lambda^2$ and of the class of the admissible,
``sufficiently small''
perturbations
 \be
 W^{(K)}(\lambda)=V^{(K)}[\Lambda(\lambda)] +
 {\rm higher\ order\ corrections}\,,\ \ \ \Lambda \gg 1
 \,.
 \ee
The
eligible
and simplified leading-order matrices $V^{(K)}(\Lambda)$
may be found displayed
in Appendix A.
Under a useful though slightly artificial
matrix-triangularity constraint
 \be
 V_{m,n}=0\,,
 \ \ \ \ \
 m \leq n=0,1,\ldots,K-1\,,
 \label{triru}
 \ee
the main part of the necessary condition
of the reality of the perturbed spectrum
has the form
 \be
 V_{m+1,m}={\cal O}(1/\Lambda^{0})\,,
 \ \ \ \ m = 0, 1, \ldots, K-2\,,
 \label{dozmatko}
 \ee
 \ben
 V_{n+2,n}={\cal O}(1/\Lambda^{1})\,,
 \ \ \ \ n = 0, 1, \ldots, K-3\,
 \een
and so on, up to
 \ben
 V_{K-1,0}={\cal O}(1/\Lambda^{K-2})\,.
 \een
More explicitly we may write
 \be
 V_{m+1,m}=a_m^{(1)}/\Lambda^{0}\,,
 \ \ \ \ m = 0, 1, \ldots, K-2\,,
 \label{vezmatko}
 \ee
 \ben
 V_{n+2,n}=a_n^{(2)}/\Lambda^{1}\,,
 \ \ \ \ n = 0, 1, \ldots, K-3\,
 \een
and so on, up to
 \ben
 V_{K-1,0}=a_0^{(K-1)}/\Lambda^{K-2}\,
 \een
where all of the constants are bounded, $a_i^{(j)}={\cal O}(1)$.
These coefficients may be arranged in
a real array which will be called fundamental matrix,
 \be
  C^{(K)}=\left [
   \begin{array}{ccccc}
   0 &1&0&\ldots&0\\
   a_0^{(1)} &0&1&\ddots&\vdots\\
   a_0^{(2)} &a_1^{(1)}&\ddots&\ddots&0\\
   \vdots &\ddots&\ddots&0&1\\
   a_0^{(K-1)} &\ldots& a_{K-3}^{(2)} &a_{K-2}^{(1)}&0
   \end{array}
    \right ]\,.
 \label{uzmatko}
 \ee
Indeed, one immediately imagines that
for the matrix elements lying inside a
``physical'' domain
${\cal D}$ the spectrum of
a given fundamental matrix can be real and non-degenerate:

\begin{lemma}
After the most elementary Kronecker-delta choice of
$a_k^{(j)}=\delta_{j,1}$, $k=0, 1, \ldots, K-1$
the set of eigenvalues $\varepsilon_n$
of matrix (\ref{uzmatko})
becomes
defined in terms of roots of
classical orthogonal Chebyschev polynomials
which are all real and non-degenerate \cite{cheby}.
\label{lekora}
\end{lemma}

 \noindent
Once the domain
${\cal D}$ is found non-empty, we may also
recall and slightly reformulate
the main result of Ref.~\cite{corridors}:

\begin{thm}
If
the spectrum of fundamental matrix (\ref{uzmatko})
is real and non-degenerate then the quantum
evolution controlled by
Hamiltonian (\ref{Aperha}) is,
in the leading-order approximation,
unitary.\label{unos}
\end{thm}
\begin{proof}
In a way explained in Appendix A, the set
of
eigenvalues $\{\varepsilon_n\}$
of the fundamental matrix of coefficients $C^{(K)}$
determines the
leading-order
bound state energies $E_n(\lambda)$
in Schr\"{o}dinger Eq.~(\ref{se17}),
 \be
 E_n(\lambda)= \sqrt{\lambda}\,\,\varepsilon_n  +
 {\rm higher\ order\ corrections}\,,\ \ \ n = 0, 1, \ldots, K-1\,.
 \ee \end{proof}
The analysis of the
role of the higher-order corrections
remains nontrivial, especially when the
spectrum of fundamental matrix
$C^{(K)}$ remains
real but degenerate. Still, the flexibility
of the
acceptable fundamental
matrices
as required by Theorem \ref{unos}
is large.

After a suitable
special choice of the parameters even some
exactly solvable models can be obtained
(cf. Lemma \ref{lekora}).
Multiple other, different choices of
admissible perturbation-determining matrices $C^{(K)}$
(yielding the other admissible spectra $\{\varepsilon_n\}$)
determine different
unfoldings
of the initial $K-$tuple degeneracy
of
conventional BH bound-state
energy spectrum.
In opposite direction,
every
admissible fundamental matrix $C^{(K)}$
marks
a path which
connects an interior
of the physical domain ${\cal D}$
with its extreme EPK boundary.
Thus, in a vicinity of that point
the family of all of the perturbation-induced unfoldings
forms
a subdomain of stability.

Our initial input choice of
perturbation (\ref{ladep})
is restricted  by
the size-suppression rules
of Eq.~(\ref{dozmatko}). This makes
the standard ${\mathbb{C}}^{(K)}$ norm of
perturbation matrices irrelevant.
In a way explained in  \cite{Ruzicka}
the
admissible matrix elements of $W^{(K)}(\lambda)$
are in fact ordered in a
hierarchy of size which merely reflects and copies the
deformations of the geometry of
the physical Hilbert space of states (in fact, its
anisotropy
increasing with the decrease of $\lambda \to 0$ \cite{lotor}).


\section{Perturbations not conserving the number of bosons\label{3p3}}

Whenever the real experimental quantum dynamics
admits the creation and annihilation
of bosons the
clear separation of the $K-$dimensional BH sub-Hamiltonians
as sampled by Eq.~(\ref{geneve}) becomes
unrealistic and artificial. At the same time the price to
pay for a breakdown of such a separation would be high,
suddenly making any constructive
study of the spectrum almost prohibitively
difficult in general.

One of the possible paths towards
the necessary mathematical simplifications
may parallel the older studies working
with the ``unperturbed''
infinite-dimensional matrix versions of
BH Hamiltonians $H_{(BH)}(\gamma)$
having the block-diagonalized
form.
This assumption could simplify the study of influence of
the
non-conservative
perturbations significantly.
In the language of physics
this assumption would mean that the ``tractable''
generalized Bose-Hubbard-type Hamiltonians
would have the form of perturbations
of the conventional $c = 0$ BH Hamiltonian
 $
 H_{(BH)}(\gamma) 
 $,
 \be
  H_{(GBH)} (\gamma,\lambda)= -{\rm i} \gamma
  \left(a_1^{\dagger}a_1 - a_2^{\dagger}a_2\right) +
  \left(a_1^{\dagger}a_2 + a_2^{\dagger}a_1\right) +
   \lambda\,H_{int}
  \left( a_1,a_1^{\dagger}, a_2,a_2^{\dagger}\right)\,.
  \label{gbhref}
 \ee
Next, using the same strategy as above we shall only
study the models with
$\gamma = 1$. We will see below that such a reduction
of generality will further simplify the technicalities
while keeping the potential physics
behind the model still sufficiently interesting.

The finite-dimensional matrices of perturbations
as studied above
will be now replaced by the full partitioned matrices
 \be
 {\cal V}_{(GBH)} (\lambda)=  \left (
 \begin{array}{cccc}
 {\cal V}^{(2,2)}(\lambda)&{\cal V}^{(2,3)}(\lambda)
 &{\cal V}^{(2,4)}(\lambda)&\ldots\\
 {\cal V}^{(3,2)}(\lambda)&{\cal V}^{(3,3)}(\lambda)
 &{\cal V}^{(3,4)}(\lambda)&\ddots \\
 {\cal V}^{(4,2)}(\lambda)&{\cal V}^{(4,3)}(\lambda)
 &{\cal V}^{(4,4)}(\lambda)&\ddots\\
 \vdots&\ddots&\ddots&\ddots
 \ea
 \right )\,. \,
   \label{refene}
 \ee
Our main task will be to gurarantee that
these matrices really
remain ``sufficiently small'', i.e., that
the related energy
spectrum would not cease to be observable and real.

\subsection{Models with one off-diagonal pair of perturbation
submatrices\label{singh}}

In an elementary nontrivial realization of
the EPK-unfolding
Hamiltonian (\ref{perha})
with an off-diagonal perturbation~(\ref{refene})
let us assume that
up to two exceptions (say,
${\cal V}^{(M,L)}(\lambda)$ and ${\cal V}^{(L,M)}(\lambda)$
with $M < L$)
all of the block-off-diagonal perturbations vanish.
Thus, one has to study the new class
of partitioned Hamiltonians
with a bounded but, otherwise,
fully general matrix of perturbations,
 \be
 \mathfrak{H}^{(M,L)}(\lambda) = H_{(BH)}(1) +
 \lambda\,{W}^{(M+L)}(\lambda) \,.
 \label{reBperha}
 \ee
with a bounded but, otherwise,
fully general matrix of ``small''
boson-boson interactions
 \be
 W^{(M+L)}(\lambda)=V^{(M+L)}[\Lambda(\lambda)] +
 {\rm higher\ order\ corrections}\,,\ \ \ \Lambda
 =1/\sqrt{|\lambda|} \gg 1
 \,.
 \ee
In such a scenario one only has to use the special
symbols for the direct sums of Jordan matrices
 $$
 {\cal J}^{(M+L)}(\eta)=
 \left [
 \begin{array}{cc}
 {J}^{(M)} (\eta)&0\\0&
 {J}^{(L)} (\eta)
 \ea
 \right ]
 $$
[cf. Eq.~(\ref{jagen})] capable to
replace
Eq.~(\ref{se17b})
of Appendix A by its partitioned form
 \be
 \left [
 \begin{array}{cc}
 {G}^{(M)} ({\Lambda})&0\\0&
 {G}^{(L)} ({\Lambda})
 \ea
 \right ]
 \,
 \left [
  {\cal J}^{(M+L)}(-E)
    + \lambda\, W^{(M+L)}(\lambda)
  \right ]\,\left | \ba  |\Psi^{(M)} \kt \\
   |\Psi^{(L)} \kt
   \ea \right \rangle
  = 0\,.
 \label{se17c}
 \ee
It contains a partitioned upgrade of the eigenvector,
 \be
 \left | \ba  |\Phi^{(M)} \kt \\
   |\Phi^{(L)} \kt
   \ea \right \rangle
  =
 \left [
 \begin{array}{cc}
 {G}^{(M)} ({\Lambda})&0\\0&
 {G}^{(L)} ({\Lambda})
 \ea
 \right ]
 \,
\left | \ba  |\Psi^{(M)} \kt \\
   |\Psi^{(L)} \kt
   \ea \right \rangle\,.
 \label{se17s}
 \ee
Working again with the two equivalent
versions of the ``small''
parameters $\lambda=1/\Lambda^2$ we
re-scale the energy,
 \be
 E=E(\lambda) =\varepsilon ({\Lambda})/{\Lambda}\,.
 \label{iskala}
 \ee
Schr\"{o}dinger equation then
acquires the obvious partitioned
structure yielding secular equation
 \be
 \det\,
 \left \{
  {\cal J}^{(M+L)}\left [-\varepsilon(\Lambda)
 \right ]
   +  M^{(M+L)}(\Lambda)
  \right \}
  = 0\,.
 \label{Bmazatko}
 \ee
In the full matrix of rescaled interactions
 \be
 M^{(M+L)}(\Lambda)=\left [
 \begin{array}{cc}
 {M}^{(M)} (\Lambda)&{A}^{(M,L)} (\Lambda)\\
 {B}^{(L,M)} (\Lambda)&
 {M}^{(L)} (\Lambda)
 \ea
 \right ]
 \label{udo}
 \ee
both of the diagonal blocks remain the same as in Appendix A.
As long as $M < L$,
the two off-diagonal non-square rescaled-interaction
submatrices
deserve an explicit display in terms of the respective
abbreviations
$W^{(P,Q)}_{mn}({\lambda})=W_{mn}={\cal O}(1)$ with superscripts
$(P,Q)=(M,L)$ or $(P,Q)=(L,M)$ yielding
 \be
 A^{(M,L)}({\Lambda})=
   \left [
   \begin{array}{ccccc}
   {\Lambda}^{-1}W_{00} &{\Lambda}^{-2}W_{01}&
   {\Lambda}^{-3}W_{02}&\ldots&{\Lambda}^{-L}W_{0,L-1}\\
   {\Lambda}^{0}W_{10} &{\Lambda}^{-1}W_{11}&
   {\Lambda}^{-2}W_{12}&\ldots&{\Lambda}^{1-L}W_{1,L-1}\\
   {\Lambda}^{1}W_{20} &{\Lambda}^{0}W_{21}&
   {\Lambda}^{-1}W_{22}&\ldots&{\Lambda}^{2-L}W_{2,L-1}\\
   \vdots &\vdots&\vdots&\ddots&\vdots\\
   {\Lambda}^{M-2}W_{M-1,0} &{\Lambda}^{M-3}
   W_{M-1,1}&{\Lambda}^{M-4}W_{M-1,2}&\ldots &{\Lambda}^{M-1-L}W_{M-1,L-1}
   \end{array}
    \right ]\,
 \label{Btatko}
 \ee
and, {\it mutatis mutandis},
the analogous formula for $B^{(L,M)}({\Lambda})$.

Next, proceeding along the same lines as in Appendix A we
accept the
boundedness assumption (\ref{ods})
and we replace the respective exact matrix-element
functions $W^{(P,Q)}$ of cut-off $\Lambda$ by their
asymptotically dominant
components ${V}^{(P,Q)}$.
The real and non-degenerate set of the
leading-order energy eigenvalues
$\varepsilon_0=\lim_{\Lambda \to \infty}\varepsilon({\Lambda})$
should be then extracted from the leading-order version of
secular equation
 \be
 \det\,
 \left [{\cal J}^{(M+L)} (-\varepsilon_0)
    +  M^{(M+L)}_0(\Lambda)
  \right ]
  = 0\,.
 \label{zBmazatko}
 \ee
It is defined in terms of the leading-order
version $M^{(M+L)}_0(\Lambda)$
of
the interaction term. Its structure
 \be
 M^{(M+L)}_0(\Lambda)=\left [
 \begin{array}{cc}
 {M}^{(M)}_0 (\Lambda)&{A}^{(M,L)}_0 (\Lambda)\\
 {B}^{(L,M)}_0 (\Lambda)&
 {M}^{(L)}_0 (\Lambda)
 \ea
 \right ]\,,
 \label{udolf}
 \ee
exhibits important simplifications, with the
diagonal blocks defined in Appendix A
(cf. Eq.~(\ref{zmatko})) and
complemented by
 \be
 A^{(M,L)}_0=
 \left [
   \begin{array}{ccccccc}
   0 &0&\ldots&0&0&\ldots&0\\
   V_{10} &0&\ddots&\vdots&\vdots&&\vdots\\
   {\Lambda}^{}V_{20} &\ddots&\ddots&0&0&\ldots&0\\
   \vdots &\ddots&V_{M-2,M-3}&0&0&\ldots&0\\
   {\Lambda}^{M-2}V_{M-1,0} &\ldots&
    {\Lambda}^{}V_{M-1,M-3} &V_{M-1,M-2}&0&\ldots&0
   \end{array}
    \right ]\,
 \label{Bzmatko}
 \ee
and by formula for $B^{(L,M)}({\Lambda})$
(with display left to the readers).
The latter two non-square matrices
as well as their two diagonal-block square-matrix partners
are all lower triangular,
dominated by their respective lowest-left-corner elements.
{\it Per analogiam} we expect
that the scope of Theorem~\ref{unos} might be extended
to cover also
the models with block-off-diagonal perturbations.
Without proof, these expectations
may be given, for Hamiltonians (\ref{reBperha}),
the following explicit formulation:



\begin{conj}
In leading-order approximation,
quantum evolution controlled by the $(M,L)$-par\-ti\-tio\-ned
Hamiltonians (\ref{reBperha}) will be unitary
provided only that
the spectrum of
the corresponding $(M,L)$-partitioned analogue of
fundamental
matrix
(\ref{uzmatko}) proves real and non-degenerate.
\label{tthm}
\end{conj}

 \noindent
At any pair of integers
$M$ and $L$ there exist only too many
fundamental matrices
$C^{(M+L)}$ with a real and non-degenerate spectrum.
One can conclude that
even in the non-conservative, partitioned-matrix
BH-type quantum systems
sampled by Eq.~(\ref{reBperha})
the
constructive guarantees of the reality of
the spectra of energies unfolded
along certain parametric paths
remain mathematically feasible.

Equation ({\ref{geneve}})
of
paragraph \ref{para2p1}
represents,
in the EPK limit $\gamma \to 1$,
our unperturbed block-diagonal Hamiltonian $H_{(BH)}(1)$.
In the present paragraph such a Hamiltonian was
considered endowed with a
most elementary perturbation
assumed to
couple just the
two basis-state sets
representing the two arbitrary non-equal
amounts of bosons, $N_1\ (=M) < N_2\ (=L)$.
The idea can easily be generalized
to the triply partitioned models with coupling of
three preselected non-equal
amounts of bosons $N_1 < N_2 < N_3$, etc.
One immediately imagines that
such steps would be too formal.
From the point of view of the dynamics of the bosons it
makes much better sense to study just
systems with more partitions but
with not too large total number of
mutually interacting bosons.


\subsection{Leading-order block-non-diagonal Hamiltonians}

Although the most straightforward modification of
the Bose-Hubbard Hamiltonian (\ref{Ham1})
has been shown to require a change of its boson-boson
interaction component
$H_{int}(a_1,a_2,a_1^\dagger,a_2^\dagger)$,
such a most obvious model-building
strategy
may have several weak points.
In the basis where the unperturbed Hamiltonian
remains block-diagonal, for example,
technical obstacles might emerge in connection
with the evaluation
of matrix elements of the perturbation. Besides that,
one must often keep in mind that
the underlying algebras of the creation and annihilation
operators can only be realized
in a strictly infinite-dimensional
Hilbert space.
For all of these reasons
an alternative strategy will be advocated and used
in what follows. In a formally less ambitious
approach
we will make an ample
use of the finite-dimensional, truncated-matrix
versions of our Hamiltonians.
In parallel, we will insist on
staying inside a
closed quantum system setup and phenomenology.
Thus, even our most complicated versions of
perturbed BH
Hamiltonians
$H_{(GBH)}(1,\lambda)$
will be required
to possess the unitarity-compatible real spectra.

We shall keep in mind that
even the most ambitious generalizations
of the interactions should
remain user-friendly.
We felt inspired by the results
outlined in Appendix A. In their light
one of the key features
of our models
will lie
in their sparse-matrix form in the
dominant-order approximation.
Although the matrices will stay
manifestly $\lambda-$dependent,
their generic structure will be
the following one,
 \be
 \mathfrak{H}^{(2+3+\ldots)}_0(\lambda)
 = \left[ \begin {array}{cc|ccc|cccc|c}
  0&1&0&0&0&0&0&0&0&\ldots
  \\
  \star&0&\star&0&0&\star&0&0&0&\ldots
  \\
  \hline
  0&0&0&1&0&0&0&0&0&\ldots
  \\
  \star&0&\star&0&1&\star&0&0&0&\ldots
  \\
  \star&\star&\star&\star&0&\star&\star&0&0&\ldots
  \\
  \hline
  0&0&0&0&0&0&1&0&0&\ldots
  \\
  0&0&0&0&0&\star&0&1&0&\ldots
  \\
  \vdots&\vdots&\vdots&\vdots&
     \vdots&\vdots&\ddots&\ddots&\ddots&\ddots
    \end {array} \right]\,.
 \label{thethe}
 \ee
Such a structure will result from a combination of the
requirement of the reality of spectra
with the
physics-dictated step-by-step choice of
the dimensions $N_1=2$, $N_2=3$, \ldots\,,
$N_{\max} =L$ responsible for the
boson-number partitioning.
After truncation, the original infinite-dimensional
matrix
$\mathfrak{H}^{(2+3+\ldots)}_0(\lambda)$
of Eq.~(\ref{thethe}) would acquire the
finite-dimensional, $K_L$ by $K_L$
matrix form
$\mathfrak{H}^{(2+3+\ldots+L)}_0(\lambda)$
where
$K_L=(L^2+L-2)/2$.

With the formal proof postponed to a forthcoming
research, our expectations
concerning the reality of spectra
have,
at present, the following form:

\begin{conj}
In leading-order approximation,
quantum evolution controlled by the truncated forms of
Hamiltonian (\ref{thethe}) will be unitary
provided only that
the spectrum of
the
fundamental
matrix
[i.e., of the corresponding $(N_1,N_2, \ldots\,,
N_{\max})$-partitioned analogue of
matrix
(\ref{uzmatko})] proves real and non-degenerate.
\label{ttthm}
\end{conj}

 \noindent
In Conjectures \ref{tthm} and \ref{ttthm} we spoke about the
leading-order approximations which were just rather vaguely specified.
Another formal weakness of these Conjectures is that at present, we
do not know how one could include, sufficiently efficiently,
the higher-order corrections,
or how one could guarantee, in a systematic manner, the required
spectral properties of the respective fundamental matrices.
This is the reason why we do not provide here the proofs
(converting
our Conjectures, in a more or less straightforward manner, into Theorems),
and why we will prefer adding a few illustrative examples in the next section.
We will shift emphasis to physics (of the most elementary systems of
``not too many'' bosons),
and we will outline there a few
more tricks helping to keep the spectra real and non-degenerate
in practical applications.

\section{Examples}

The main technical challenge connected with the study of
models of preceding section
may be found formulated in Conjecture \ref{ttthm}: One needs to
guarantee that
the spectrum of a
fundamental
matrix is real and non-degenerate.
Unfortunately, even at the smallest possible
partition
dimensions
$N_1\ (=M) =2$ and $N_2\ (=L) =3$
the
corresponding
secular polynomial
will be a polynomial of the fifth degree in the energy.
For this reason, the
necessary specification of its admissible, dynamics-determining coefficients
[i.e., of the unitarity-compatible matrix elements of
perturbations $\lambda W^{(M+L)}(\lambda)$
in Eq.~(\ref{reBperha})]
seems to be a purely numerical task.
Now, let us show that such a task can be
reduced to a sequence of easier steps.

\subsection{Models with the bosonic pairs coupled to triples}

\subsubsection{The most elementary case}

In the above-mentioned
most elementary generalized-Bose-Hubbard example with
$N_1\ (=M) =2$ and $N_2\ (=L) =3$
one has to consider
Hamiltonian
 \be
 \mathfrak{H}^{(2+3)}(\lambda)=
 \left[ \begin {array}{cc|ccc}
  0&1&0&0&0
  \\ \star&0&\star&0&0
   \\
   \hline
 0&0&0&1&0
 \\ \star&0&\star&0&1
 \\ \star&\star&\star&\star&0\end {array} \right]\,
 \label{forha}
 \ee
with the potentially non-vanishing
matrix elements marked by stars $\star$.
Up to scalar factor $1/\sqrt{\lambda}$ this Hamiltonian
should be
isospectral with the fundamental
$(2,3)$-partitioned matrix
 \be
 C^{(2+3)}= \left[ \begin {array}{cc|ccc}
  0&1&0&0&0
  \\a&0&b&0&0
   \\
   \hline
 0&0&0&1&0
 \\c&0&d&0&1
 \\e&f&g&h&0\end {array} \right]\,.
 \label{predeal}
 \ee
This is an ansatz which
varies with eight free parameters
representing the dominant components of the
interaction. Now,
in spite of the quintic-polynomial
(i.e., exactly unsolvable)
nature of the
corresponding secular equation,
what is to be sought are the criteria
of admissibility of these parameters in the unitary
evolution regime.


Intuitively, the existence of the octuplet of free parameters
might prove sufficient for an adjustment to
any input quintuplet of
eigenvalues $\varepsilon_n$. Quickly, one finds that such a recipe is not viable.
Safer conclusions can only
result from a construction
based on
the explicit form
 $$
 \det ( C^{(2+3)}- \varepsilon \,I)=0\,
 $$
of the underlying quintic-polynomial secular equation.

A way to an efficient simplification of the problem
lies in the fact that
in the secular
polynomial the
coefficient at the
second power of
$\varepsilon$
is equal to $g$. Thus,
the term drops out after one selects
$g=0$. This immediately implies that when we further set $e=0$,
the five roots of the secular polynomial
acquire the following elementary
form
 \be
 0, \pm \frac{\sqrt{a+d+h\pm \sqrt{4b(c+f)+(d+h-a)^2)}}}{\sqrt{2}}\,.
 \label{formuf}
 \ee
Due to the elementary nature of the model with $M=2$ and $L=3$
(where the selection of $g=e=0$ was ``obvious'') the problem is solved.
Unfortunately, the similar easy simplifying selection ceases to be
available whenever $M>2$ and/or $L>3$.

Let us now describe an alternative, more robust recipe
by which the existence of the necessary fundamental matrices with real spectra
could be proved, in a systematic iterative manner, even at the larger $M>2$ and/or $L>3$.
In the first step let us return to the benchmark model (\ref{formuf})
and let us make a ``wrong'' choice having set all of the
remaining parameters equal to one.
Naturally, we get an unsatisfactory answer because the
spectrum appears real but partially degenerate.
Still, although the trial and error guesswork did not work,
a subsequent return to the $g \neq 0$ tentative
amendment using $f=1/100$ and $g=1/2$
already appears to serve the purpose.
What is obtained is the
numerical (i.e., approximate)
quintuplet of well separated real
eigenvalues
 $$
  \{ -1.5118, -0.4630, 0.0000, 0.4630, 1.5118\}\,.
  $$
  %
  %
This offers an alternative proof of
the existence of at least one
fundamental matrix with the required properties.

Incidentally, as long as the model with $K_L=5$
is not yet too large,
the alternative proof
based on the ``wrong'' choice of $g \neq 0$
can still be made non-numerical. Indeed,
besides the obvious exact root $\varepsilon_0=0$,
also the other four
non-vanishing ones can be given
the closed form
 $$\varepsilon=\varepsilon_{\pm,\pm}=
(\pm \sqrt {390} \pm \sqrt {110})/20\,$$
which parallels formula (\ref{formuf}).

\subsubsection{Boundary $\partial{\cal D}$ of the
corridor of unitarity\label{secA}}

In preceding subsection we demonstrated that the
corridor
${\cal D}={\cal D}^{(2+3)}$ of the
admissible parameters
of unitary unfoldings of
the twice
degenerate Bose-Hubbard EP = EP2 + EP3 extreme is a
nonempty
domain.
This opens a number of new questions
concerning the shape and properties of
the boundary $\partial{\cal D}^{(2+3)}$ of stability.

In a preparatory step we notice that in the original
``eight-star'' Hamiltonian
(\ref{forha})
only six stars stand for the dominant
components of the
interaction
[i.e., for the ${\cal O}(\Lambda^0)$
matrix elements of $V$,
see the first line of Eq.~(\ref{dozmatko})].
For this reason we will omit the study of the role of the
two next-order  ${\cal O}(\Lambda^{-1})$ contributions
[cf. the second line in Eq.~(\ref{dozmatko})] and
we will set again, for methodical reasons,
$e=g=0$ in
our five by five fundamental matrix  (\ref{predeal}).
Then, the related
secular polynomial
 \be
 P={{\it {z}}}^{5}- \left( a+d+h \right) {\it {z}}^3
 + \left( a(d+h)-b(f+c) \right)  {\it {z}}\,
 \label{petak}
 \ee
defines,
via the condition of
reality and non-degeneracy of all
of its five roots (\ref{formuf}),
the whole six-dimensional physical domain
${\cal D}^{(2+3)}$ of parameters
$a$, $b$, $c$, $d$, $f$ and $h$ of
dynamical relevance.

The ``most  common''
parts of  boundary $\partial{\cal D}^{(2+3)}$
could be now identified with separate
submanifolds of pairwise EP2 mergers
of the roots. The boundary may also contain
the lower-dimensional
parts supporting the
higher-order mergers, up to
the most interesting EP5 extreme if any.

As long as one of the roots
of the secular polynomial
is constant, ${z}_0=0$,
the localization of the EP5 boundary manifold
$\partial{\cal D}^{(2+3)}$
is facilitated. Indeed, the secular polynomial
must degenerate
to monomial, $P^{(EP5)}={z}^5$.
In the light of Eqs.~(\ref{formuf}) (\ref{petak}) this
leads to two necessary-condition
equations yielding the $b \neq 0$ ``solution A'',
 \be
 h=-a-d\,,\ \ \ \ f=-a^2/b-c\,
 \label{sola}
 \ee
and its $b=0$ complement, ``solution B'',
 \be
 a=b=0\,,\ \ \ h=-d\,.
 \label{solbe}
 \ee
In the former case ``A'', condition (\ref{sola}) is also
sufficient.
This means that
condition (\ref{solbe}) of case ``B''
merely leads to the less singular types of degeneracy.
Its analysis may be found transferred to Appendix B.

Once we return to
the spectral-degeneracy limit of type ``A'' our
fundamental matrix
$C^{(2+3)}_{(A)}$ is really easily shown to
satisfy
the EP5 analogue
of Eq.~(\ref{JBK}),
  \be
   C^{(2+3)}_{(A)}
 =Q^{}_{(A)}\,
 J^{(5)}
 (0)\, \left [{Q^{}_{(A)}}
 \right ]^{-1}\,.
   \label{JaBKa}
 \ee
As long as
such a limiting matrix
still contains four freely variable real parameters
$a,\, b,\, c$ and $d$,
also the corresponding
part of the
boundary $\partial{\cal D}^{(2+3)}$
is a four-dimensional manifold.
For its compact description
it makes sense to
abbreviate $-a^2/b-c=F(a,b,c)=F$.
Then it is easy to display the transition matrix
 $$
 Q^{}_{(A)}=\left[ \begin {array}{ccccc}
  -{F}b&0&a&0&1\\{}0&-{F}b&0&a&0
  \\{}a{F}&0&c&0&0
  \\{}0&a{F}&0&c&0
  \\{}- \left( {a}^{2}+{F}b+ad \right) {F}&0&
  -({ {2\,ca+cd+{a}^{3}/b}})
  &0&0\end {array} \right]
 \,
 $$
and to prove that it is invertible
since $\det Q^{}_{(A)} = -{F}^{5}{b}^{2}$.

\subsection{Bosonic pairs and triples coupled to
quadruples ($L=4$, $K_L=9$)}

Let the dominant-order Hamiltonian have the following
sparse, nine by nine matrix form
 \be
 \mathfrak{H}^{(2+3+4)}_0(\lambda)
 = \left[ \begin {array}{cc|ccc|cccc}
  0&1&0&0&0&0&0&0&0
  \\
  \star&0&\star&0&0&\star&0&0&0
  \\
  \hline
  0&0&0&1&0&0&0&0&0
  \\
  \star&0&\star&0&1&\star&0&0&0
  \\
  \star&\star&\star&\star&0&\star&\star&0&0
  \\
  \hline
  0&0&0&0&0&0&1&0&0
  \\
  0&0&\star&0&0&\star&0&1&0
  \\
  \star&0&\star&\star&0&\star&\star&0&1
  \\
  \star&\star&\star&\star&\star&\star&\star&\star&0
      \end {array} \right]\,.
 \label{rethethe}
 \ee
The stars mark the $\lambda-$dependent
matrix elements which should again characterize the
leading-order components of the admissible,
real-spectrum-supporting
boson-boson interaction $\lambda\,{\cal V}^{(2+3+4)}$.

After an appropriate rescaling made in the spirit
of Eq.~(\ref{vezmatko}),
these components become proportional to the
29 relevant parameters entering fundamental matrix
with the same sparse-matrix structure.
For pedagogical reasons we will still omit
all of the representatives of the subdominant
perturbations (for example, in Eq.~(\ref{dozmatko})
we would only keep the elements of the first row).
In this way our candidate for the fundamental matrix will
have the following simplified, 17-parametric form
 \be
 C^{(2+3+4)}=\left[ \begin {array}{cc|ccc|cccc}
  0&1&0&0&0&0&0&0&0
 \\
 {}a&0&b&0&0&c&0&0&0
 \\
 \hline
 {}0&0&0&1&0&0&0
&0&0
\\
{}d&0&e&0&1&f&0&0&0
\\
{}0&h&0&j&0
&0&l&0&0
\\
\hline
{}0&0&0&0&0&0&1&0&0
\\
{}o&0&m
&0&0&n&0&1&0
\\
{}0&u&0&q&0&0&s&0&1
\\
{}0
&0&0&0&x&0&0&{\it \omega}&0
\end {array} \right]
\label{nepredeal}
 \ee
This matrix is an
$(2+3+4)$ analogue of
its $(2+3)$ partitioned predecessor (\ref{predeal})
as well as of the even simpler,
unpartitioned
matrix (\ref{uzmatko}).
Thus, in leading-order approximation,
quantum evolution controlled by
Hamiltonian (\ref{rethethe}) will be unitary
provided only that
the spectrum of
fundamental
matrix (\ref{nepredeal}) proves real and non-degenerate.

Our final task is to prove the existence of the
unitarity-supporting
boson-boson interactions:

\begin{lemma}
There exists a non-empty domain of
matrix elements in (\ref{nepredeal})
yielding the real and non-degenerate spectrum
of fundamental matrix.
\label{lemadva}
\end{lemma}
\begin{proof}
The process of proof will be similar to the preceding case.
In the first step we set all parameters equal to one
an obtain the elementary secular polynomial
$
{{\it {z}}}^{9}-6\,{{\it {z}}}^{7}+3\,{{\it {z}}}^{5}
$ with the nine real roots
 $$
 \{0, 0, 0, 0, 0, -(3+6^{1/2})^{1/2},
  (3+6^{1/2})^{1/2}, -(3-6^{1/2})^{1/2}, (3-6^{1/2})^{1/2}\}\,,
  $$
i.e., numerically,
 $$
 \{
 \pm 2.334414218, \pm 0.7419637843, 0., 0., 0., 0., 0.
 \}.
 $$
The quintuple degeneracy
was weakened by a modification of $b=1+v$.
This yielded a modified secular polynomial
which was linear in $v$. The
behavior of the easily obtained curve $v=v({z})$
near the origin indicated a weakening of the degeneracy
for $v \in (-1/4,0)$, so we choose
$v=-1/10$ and obtained the amended spectrum
 $$\{
\pm 2.327224413, \pm 0.7154619694, \pm 0.2685902108, 0., 0., 0.
\}\,.$$
With the degeneracy
reduced  to three we iterated the process
and added a new auxiliary variable $w$
to elements $a$, $l$ and $u$.
Having evaluated the (this time, double-branched) function
$w({z})$. As long as his function appeared to have
two branches,
 $$
 w_1(z)=(-{\frac {1}{10}}+{\frac {319}{200}}
 {{\it {z}}}^{2}+O \left( {{\it {z}}}^{4} \right) )
 $$
and
 $$
 w_2(z)=(-{{\it {z}}}^{2}+O \left( {{\it {z}}}^{4} \right) )
 $$
forming a small circle below the real line at small ${z}$,
we concluded that the small negative value will work.
Indeed, with $w=-1/100$ we obtained
 $$\{
\pm 2.325373957, \pm 0.7112721146, \pm 0.2586335768, \pm 0.09917969302, 0.
\}$$ i.e., the sample spectrum we needed for the proof.
\end{proof}

 \noindent
We may conclude that the
introduction of the new non-vanishing matrix elements
safely removed the degeneracy while still keeping the
other eight roots almost unchanged.
Adding that the latter spectrum
yields,
after a premultiplication by factor $\sqrt{\lambda}$, the
ultimate bound-state energies.

\section{Conclusions}

The main mathematical
message delivered by our
present paper is that
there are not too many really deep conceptual
differences between the models with the conventional,
simple exceptional points
(sampled by Eq.~(\ref{JBK}) + (\ref{JBKz}))
and the generalized models with the much less usual,
degenerate exceptional points
(sampled, e.g., by Eqs.~(\ref{JB}) + (\ref{jagen})).
Most importantly, what is shared
is the correlation between
the natural physical requirement of the unitarity of the model
(i.e., of the reality of the spectrum)
and the highly artificial-looking mathematical requirements
by which
the matrix elements of the
underlying ``admissible'' perturbation matrices
must exhibit an
ordering in size
as sampled by Eq.~(\ref{dozmatko}).

Our present results may be read as a climax of recent developments
in the field.
First of all, the formal,
purely algebraic-geometry origin of the necessity of the
apparently strongly
counterintuitive hierarchy of the matrix elements
of admissible perturbation matrices
may be found described,
in pedagogical detail,
in Section Nr. IV of Ref.~\cite{corridors}. Secondly,
in the context of physics of closed quantum systems,
a deeper clarification of the apparent paradox
of irrelevance of the conventional norms of matrices
of perturbations
should be sought
in Ref.~\cite{admissible}. The conclusion is
that as long as our quantum system of interest lives, by assumption,
in the vicinity of its
loss-of-the-unitarity boundary $\partial {\cal D}$, the
enormous differences between the
influence of
separate matrix elements of the interaction just reflect the
deformations of
the geometry
of the physical Hilbert space. Indeed, this geometry becomes,
near $\partial {\cal D}$, increasingly anisotropic
(for more details see also an extensive commentary
on this topic in \cite{lotor}).

Besides these mathematical results,
our attention was attracted by
the descriptive features of
the conventional BH Hamiltonian of Eq.~(\ref{Ham1}).
From this point of view the main
message delivered by our
present paper was threefold. Firstly,
in reference to the extensive study \cite{Uwe}
of the role of the specific
boson-boson-interaction perturbations $c\,H_{int}$
we emphasized that near the EPK
extreme of boundary $\partial {\cal D}$,
such a perturbation would
be out of our present closed-system-oriented
interest because, due to its complex energies, it
only describes physical reality
in the traditional open, unstable quantum system setting.

Secondly, guided by the results
presented in paragraph \ref{3p2}
and in Appendix A we
found a new domain of applicability of the
closed quantum system philosophy
under assumption that the number of bosons
is constant.
For the purpose
we recommended the replacement of
conventional, ``too large''
boson-boson interaction term $c\,H_{int}$
by its ``sufficiently small''
boson-boson interaction
amendment as
specified by Theorem \ref{unos}.

Thirdly, we observed that
the mathematics behind the
latter amendment need not necessarily remain restricted to the
perturbations which commute with
the boson number operator $\widehat{N}$ of Eq.~(\ref{Num1}).
Subsequently we imagined that such an innocent-looking
extension of mathematics leads to an enormous extension
of the descriptive phenomenological
capacity of the model. Indeed, in the dominant-order
approximation we managed to reduce the
problem of the guarantee of the unitarity
of evolution of the system in question
(i.e., of the reality of the $\lambda-$dependent
bound-state energies) to the purely mathematical
analysis of spectra of certain auxiliary, sparse and
$\lambda-$independent ``fundamental'' matrices $C$.

The feasibility of application of the latter criterion
was finally demonstrated via the first two
nontrivial illustrative examples
in which the number of bosons was not conserved
[cf. their respective partitioned-matrix
Hamiltonians in Eqs.~(\ref{forha}) and (\ref{rethethe})]
and in which the exceptional-point singularities
acquired the extremely interesting degenerate-EP structures
of the form of the direct sums E2+E3 and E2+E3+E4, respectively.

Marginally, it seems worth adding that
the numbers of the relevant
variable parameters which were ``multi-indexing'' the
corresponding five- and nine-dimensional
matrices of  perturbations
were 8 and 27, respectively.
Still, after a number of trial and error ``experiments''
we eliminated some of the less relevant variables
and, ultimately, we managed to prove
the non-emptiness of the respective ``physical'',
unitarity-supporting
domains ${\cal D}$ of
parameters.
Moreover,
last but not least we accompanied the latter
proofs of existence also by
an explicit localization and exceptional-point
classification of a few parts of the
end-of-unitarity quantum phase transition
boundary $\partial {\cal D}^{(2+3)}$.

\section*{Acknowledgments}

The author acknowledges the financial support from the
Excellence project P\v{r}F UHK 2020.

\section*{Funding statement}

The author's work was supported by the Nuclear Physics Institute
of the CAS, and by the University of Hradec Kr\'alov\'e.

\newpage

\section*{Appendix A: Conservative BH perturbations (\ref{ladep})
and the necessary and sufficient condition
of the reality of spectrum}

The demonstration of the
reality of the spectra as described
in paper
\cite{corridors}
is strongly model-dependent and
hardly applicable
in the present general BH-related setup.
Let us propose and describe a different,
shorter version
of the method, therefore.

Firstly, in our perturbed Schr\"{o}dinger equation (\ref{se17})
let us preselect any ground- or excited-state subscript
$n=n_0$. Then, in a shorthand notation let us
just remember
this information and simplify
$|\Psi_{n_0}(\lambda)\kt \to |\Psi(\lambda)\kt$
and $E_{n_0}(\lambda) \to E_{}(\lambda)$.
Secondly, let us introduce a large, $\lambda-$dependent
cut-off parameter
${\Lambda}={\Lambda}(\lambda)=1/\sqrt{\lambda}$
entering a one-parametric auxiliary diagonal
matrix ${G}^{(K)} ({\Lambda})$
with ``increasingly large''
elements ${G}_{kk}^{(K)} ({\Lambda})={\Lambda}^k$,
$k=0,1,\ldots,K-1$.
Next, let us re-write the perturbed
Schr\"{o}dinger equation (\ref{se17})
in a preconditioned $K$ by $K$ matrix form
 \be
 {G}^{(K)} ({\Lambda})
 \left [  J^{(K)}(0)
  + \lambda\, W^{(K)}(\lambda) - E(\lambda) I
  \right ]\,\left[{G}^{(K)}  ({\Lambda}(\lambda))\right ]^{-1}
 \,|\Phi (\lambda) \kt= 0\,,
 \ \ \ \ |\Phi (\lambda) \kt={G}^{(K)} ({\Lambda})\, |\Psi(\lambda)  \kt
 \,.
 \label{se17b}
 \ee
Finally let us re-scale the energy
$E(\lambda) =\varepsilon ({\Lambda})/{\Lambda}\,$
and transform our initial, conventional secular equation
 \be
 \det\,\left \{
  {G}^{(K)} [{\Lambda}(\lambda)]
 \left [  J^{(K)}(0)
  + {\lambda}\, W^{(K)}({\lambda}) - E({\lambda}) I^{(K)} \right ]
  \,\left[{G}^{(K)}  ({\Lambda}(\lambda))\right ]^{-1}
  \right \}
  = 0
 \label{honzatko}
 \ee
into the following equivalent equation
 \be
 \det\,
 \left [  J^{(K)}(0)
  + M^{(K)}({\Lambda}) - \varepsilon({\Lambda}) I^{(K)} \right ]
  = 0\,.
 \label{mazatko}
 \ee
The re-scaled
matrix of perturbations is defined in terms of elements
$W_{mn}=W^{(K)}_{mn}({\lambda})$ multiplied by powers of our
cut-off-resembling large parameter $\Lambda=1/\sqrt{\lambda}$,
 \be
 M^{(K)}({\Lambda})=
   \left [
   \begin{array}{ccccc}
   {\Lambda}^{-1}W_{00} &{\Lambda}^{-2}
   W_{01}&{\Lambda}^{-3}W_{02}&\ldots&{\Lambda}^{-K}W_{0,K-1}
   \\
   {\Lambda}^{0}W_{10} &{\Lambda}^{-1}W_{11}&
   {\Lambda}^{-2}W_{12}&\ldots&{\Lambda}^{1-K}W_{1,K-1}
   \\
   {\Lambda}^{1}W_{20} &{\Lambda}^{0}W_{21}&
   {\Lambda}^{-1}W_{22}&\ldots&{\Lambda}^{2-K}W_{2,K-1}\\
   \vdots &\vdots&\vdots&\ddots&\vdots\\
   {\Lambda}^{K-2}W_{K-1,0} &{\Lambda}^{K-3}W_{K-1,1}&
   {\Lambda}^{K-4}W_{K-1,2}&\ldots &{\Lambda}^{-1}W_{K-1,K-1}
   \end{array}
    \right ]\,.
 \label{tatko}
 \ee
Under the boundedness assumption (\ref{ods})
and using a leading-order-coefficient simplification
 $$
 W^{(K)}_{mn}({\Lambda})=V_{mn}\Lambda^{constant}+ {\rm corrections}\,,
 \ \ \ \ m,n=0, 1, \ldots, K-1
 $$
we get the
lower triangular leading-order matrix of perturbations
 \be
 M^{(K)}_0({\Lambda})=
 \left [
   \begin{array}{ccccc}
   0 &0&\ldots&0&0\\
   V_{10} &0&\ddots&\vdots&\vdots\\
   {\Lambda}^{}V_{20} &\ddots&\ddots&0&0\\
   \vdots &\ddots&V_{K-2,K-3}&0&0\\
   {\Lambda}^{K-2}V_{K-1,0} &\ldots& {\Lambda}^{}V_{K-1,K-3} &V_{K-1,K-2}&0
   \end{array}
    \right ]\,.
 \label{zmatko}
 \ee
The insertion
of this matrix in Eq.~(\ref{mazatko})
with
$\varepsilon({\Lambda})=\varepsilon_0+ {\rm corrections}$
yields
the leading-order secular equation
 \be
 \det\,
 \left [
   \begin{array}{ccccc}
   -\varepsilon_0 &1&0&\ldots&0\\
   V_{10} &-\varepsilon_0&1&\ddots&\vdots\\
   {\Lambda}^{}V_{20} &\ddots&\ddots&\ddots&0\\
   \vdots &\ddots&V_{K-2,K-3}&-\varepsilon_0&1\\
   {\Lambda}^{K-2}V_{K-1,0} &\ldots&
    {\Lambda}^{}V_{K-1,K-3} &V_{K-1,K-2}&-\varepsilon_0
   \end{array}
    \right ]\,
    =0\,.
 \label{xzmatko}
 \ee
After its
systematic analysis as sampled in both Refs.~\cite{admissible}
and \cite{corridors}
one comes to the conclusion that
the spectrum cannot be real unless one accepts the
assumption that
in the above-mentioned conventional Hilbert space
${\cal K}={\cal C}^K$
also  matrix (\ref{zmatko})
is kept bounded,
 \be
  M^{(K)}_0 \in {\cal B}({\cal K})\,.
  \label{odsb}
  \ee
This type of constraint was rendered possible by
the above-mentioned requirement
(\ref{ladep}) of an explicit variability of the separate
matrix elements with the cut-off or strength
of perturbation $\lambda=1/\Lambda^2$.
Thus,
besides the matrix-triangularity rule (\ref{triru}) as used also in Eq.~(\ref{zmatko}),
the reality of the perturbed spectrum
is, in general, guaranteed by
equations (\ref{dozmatko})
and (\ref{vezmatko}).
These equations represent the
two alternative versions of the necessary
condition of the unitarity of the evolution.
In the present notation this means that
the underlying fundamental matrix (\ref{uzmatko})
must have a real and discrete spectrum.
Then, with all of its $K(K-1)/2$ variable parameters
such a matrix
specifies
an admissible perturbation.
In this sense, any such a matrix
could be interpreted as a definition of
one of the eligible paths through a corridor of
unitary
unfolding of the EPK spectral degeneracy.

\section*{Appendix B. A few non-EP5
components of the boundary of stability $\partial{\cal D}^{(2+3)}$
}

Let us recall
condition (\ref{solbe}) in case ``B''
of subsection \ref{secA}. It is easily shown to
lead to factorization
  \be
   C^{(2+3)}_{(B)}
 =Q^{}_{(B)}\,
 J_{(B)}
 \, \left [Q^{}_{(B)}
 \right ]^{-1}\,
   \label{JbBKb}
 \ee
where
 $
 J_{(B)}={\cal J}^{(4+1)} (0)$.
This means that we have to deal with
the mere EP4 boundary of the physical parametric domain
${\cal D}$. Indeed,
using abbreviation $f+c=\alpha=\alpha(f,c)$
and
assumption $c \neq 0$
the related three-parametric transition matrix
with  $\det Q^{}_{(B)} = { { \alpha ^{3}}/{c}}$
looks particularly elementary,
 $$
 Q^{}_{(B)}= \left[ \begin {array}{ccccc}
 0&0&1&-{c}^{-1}&-{c}^{-1}
\\{}0&0&0&1&0\\{}\alpha&0&0&0&0
\\{}0&\alpha&0&0&0\\{}- \alpha
 d&0&f&1&1\end {array} \right]
 \,.
 $$
In the singular limit $c \to 0$ with $f \neq 0$,
remarkably enough, the transition matrix
remains regular,
 $$
 Q^{}_{(B,0)}= \left[ \begin {array}{ccccc}
  0&0&1&-{f}^{-1}&-{f}^{-1}
 \\{}0&0&0&1&0
 \\{}f&0&0&0&0
\\{}0&f&0&0&0
\\{}-fd&0&f&0&0
\end {array} \right]\,,
 \,.
 $$
with  $\det Q^{}_{(B,0)} = { { \alpha ^{3}}/{c}}$.

Obviously, the subsequent singular
limit of $f \to 0$
still deserves a separate treatment.
In Eq.~(\ref{JbBKb}) this leads,  by direct computations, to
$
 J_{(B)}=J_{(B,0,0)}={\cal J}^{(3+2)} (0)$
and
  $$
 Q^{}_{(B)}= Q^{}_{(B,0,0)}=
 \left[ \begin {array}{ccccc} 0&1&0&-1&0
 \\{}0&0&1&0&-1
 \\{}d&0&1&0&0
 \\{}0&d&0&0&0
 \\{}-{d}^{2}&0&0&0&0
 \end {array} \right]\,.
 $$
This is an invertible matrix since $\det Q^{}_{(B,0,0)} = { { d^{3}}}$.

In the singular limit of $d \to 0$
the fundamental matrix itself acquires an EP2+EP3 form.
It is probably worth adding that its conversion into
a reordered canonical form EP3+EP2
still requires a nontrivial, non-permutation parameter-free
transition matrix
  $$
 Q^{}_{(B)}= Q^{}_{(B,0,0,0)}=
 \left[
 \begin {array}{c|cc|cc}
 0&1&0&-1&0
 \\
 {}0&0&1&0&-1
 \\
 \hline
 {}1&0&0&0&0
 \\
 \hline
 {}0&1&0&0&0
 \\
 {}0&0&1&0&0
 \end {array} \right]\,
 $$
with unit determinant.

\end{document}